\documentclass[12pt]{article}
\usepackage{amssymb}
\usepackage{epsfig}
\usepackage{latexsym}
\usepackage{graphicx}
\usepackage{tikz}
\usepackage{scalefnt}
\usepackage[textwidth=165mm,textheight=250mm]{geometry}
\usepackage{caption}
\usepackage{url}
\newcommand{\ds}{\displaystyle}

\newcommand{\blackbox}{\mbox{}\\  \hspace*{\fill} \rule{3mm}{3mm}}
\newcommand{\bbm}{\begin{boldmath}}
\newcommand{\ebm}{\end{boldmath}}

\newcommand{\tikzpicturefont}{0.5}
\newtheorem{theo}{Theorem}

\newtheorem{lem}{Lemma}[section]

\newtheorem{defy}{Definition}[section]

\newtheorem{exam}{Example}
   {\rule{\textwidth}{0.3mm}\newline%
    \begin{exam}\begin{em}\mbox{}\\}%
   {\mbox{}\\ \rule{\textwidth}{0.3mm}\end{em}\end{exam}}%
\newenvironment{proof}%
    {{\bf Proof}}%
    {\blackbox}%
\newtheorem{remk}{{\bf Remark:}}[section]
   {\begin{remk}\begin{normalshape}\mbox{}}%
   {\hfill  \end{normalshape}\end{remk}}%

\tikzstyle{level 1}=[sibling distance=40mm]
\tikzstyle{level 2}=[sibling distance=20mm]
\tikzstyle{level 3}=[sibling distance=10mm]
\tikzstyle{level 4}=[sibling distance=5mm]
\begin{document}
\begin{center}
{\Large {\bf An Algorithm to Solve the}}\\
{\Large {\bf Equal-Sum-Product Problem}}\\
\vspace{0.5cm}
{\bf M. A. Nyblom and C. D. Evans}\\
%{\bf RMIT University, GPO Box 2476V}\\
%{\bf Melbourne, Victoria 3001, Australia}\\
\end{center}
\begin{abstract}
A recursive algorithm is constructed which finds all solutions to a class of Diophantine equations connected to the 
problem of determining ordered $n$-tuples of positive integers satisfying the property that their sum is equal to their product. An 
examination of the use of Binary Search Trees in implementing the algorithm into a working program is given. In addition an application of the algorithm 
for searching possible extra exceptional values of the equal-sum-product problem is explored after demonstrating a link between these numbers and the
 Sophie Germain primes.
\end{abstract}
\section{Introduction}
Suppose we are asked to consider the following three arithmetic identities
\[
2+2=4\mbox{ ,}\hspace{1.7cm}1+2+3=6\mbox{ ,}\hspace{1.7cm}1+1+2+2+2=8\mbox{ .}
\]
What can we say is a feature common to each of the three identities? Looking at the second equality we might first think that we are dealing with the property of perfect
numbers, namely that the number $6$ is equal to the sum of it's proper divisors $1$, $2$ and $3$, but neither $4$ or $8$ are perfect numbers. However if each of the right-hand sides are expressed as products in the following manner
\[
2+2=2\cdot2\mbox{ ,}\hspace{1.7cm}1+2+3=1\cdot2\cdot3\mbox{ ,}\hspace{1.7cm}1+1+2+2+2=1\cdot1\cdot2\cdot2\cdot2\mbox{ ,}
\]
then we can see at once that the original identities express the fact that three sets of numbers $\{ 2,2\}$, $\{ 1,2,3\}$ and $\{1,1,2,2,2\}$, have the property that the sum of their elements is equal to their respective products. In view of these sets, one is naturally drawn to question whether it is possible to find for each integer $n\geq 2$, all sets of $n$ positive integers having the equal-sum-product property. We shall refer to this problem as the Equal-Sum-Product problem in $n$ variables, or the ESP-Problem in $n$ variables for short. Determining a set of $n$ positive integers satisfying the equal-sum-product property is equivalent to solving the following equation
\begin{equation}
x_{1}x_{2}\cdots x_{n}=x_{1}+x_{2}+\cdots +x_{n}\mbox{ ,}\label{eq:2ex}
\end{equation}
in positive integers $x_{i}$, where without loss of generality we may assume $x_{1}\leq x_{2}\leq\cdots\leq x_{n}$. Equations in which only positive integer solutions are sought are referred to as Diophantine equations, after the mathematician Diophantos who lived in Alexandria around 300 A.D. When examining Diophantine equations the following three questions naturally arise: Does the equation have a solution? Are there only finitely many solutions? Is it possible to determine all solutions? In the case of the Diophantine equation in (\ref{eq:2ex}), the first two questions have been independently answered in the affirmative by M. W. Ecker in \cite{ecker} and by L. Kurlandchik and A. Nowicki in \cite{kurland}. In particular the first question was easily answered via the observation that for any $n\geq 2$ the $n$-tuple $(\underbrace{1,1,\ldots ,1}_{(n-2)1's},2,n)$ is clearly a solution of~(\ref{eq:2ex}). While for the second question, the finiteness of solutions followed from a demonstration of a boundedness result, namely that for each $n\geq 2$ when (\ref{eq:2ex}) is satisfied by an $n$-tuple $(x_{1},x_{2},\ldots ,x_{n})$ then the largest component can be at most $n$, thus yielding an extremely large but finite upper bound of $n^{n}$ for the number of solutions. Despite the vast search space of solutions, we will show in this paper that the third question can also be answered in the affirmative, in the case of (\ref{eq:2ex}) by constructing an algorithm which generates all solutions to the ESP-Problem in $n$ variables. To help construct the algorithm, it first will be necessary to determine the structure of the solution set $S(n)$. In particular, by making use of existing results again found in \cite{ecker}, \cite{kurland} 
 we will show that $\ds S(n)=\cup_{r=2}^{\lfloor\log_{2}n\rfloor +1}S_{r}(n)$, where $S_{r}(n)$ is the set of solutions of (\ref{eq:2ex}), having precisely $n-r$ unit and $r$ non-unit components, that is of the form $(\underbrace{1,1,\ldots ,1,}_{(n-r)1's}x_{1},\dots ,x_{r})$. For notational convenience the solutions contained in $S_{r}(n)$ will be denoted by $(x_{1},x_{2},\ldots,x_{r};n-r)$, with $x_{i}\geq 2$.\\
\\
From this description of the solution set $S(n)$ the basic function of the algorithm can be explained. Fundamentally the algorithm will be recursive in nature, as it shall generate solutions in each set $S_{r}(n)$ for $r=3,\ldots ,\lfloor\log_{2}n\rfloor +1$, 
 from those found in $S_{r-1}(r+j)$, for $j=2^{r-2}-r,\ldots ,\left\lfloor \frac{n-3r+2}{2}\right\rfloor$. Specifically we will show that to generate each set $S_{r}(n)$, it will suffice to examine the solutions $(x_{1},x_{2},\ldots ,x_{r-1};j+1)\in S_{r-1}(r+j)$, for $j=2^{r-2}-r,\ldots ,\left\lfloor \frac{n-3r+2}{2}\right\rfloor$ which satisfy the divisibility condition $(x_{1}+x_{2}+\cdots +x_{r-1}-1)|(x_{1}+x_{2}+\cdots +x_{r-1}+n-r)$, and construct $S_{r}(n)$ as the set of elements of the form
$(x_{1},x_{2},\ldots ,x_{r-1},w;n-r)$, where $w:=(x_{1}+x_{2}+\cdots +x_{r-1}+n-r)/(x_{1}+x_{2}+\cdots +x_{r-1}-1)$. Given that the set $S_{2}(n)$ can be determined from an explicit formula in (\ref{eq:1ex}), we can see that when $r-1>2$ the algorithm must repeatedly apply the recursive procedure to generate each set $S_{r-1}(r+j)$. This presents us with a problem in that before each set $S_{r}(n)$ can be constructed, one will first have to keep track of a specific sequence of intermediary sets through descending values of $r$ to $r=2$, and then determine an associated  group of ``base`` sets $S_{2}(\cdot)$ via (\ref{eq:1ex}). Secondly one must then apply the respective divisibility tests in the reverse sequence order to construct each of the intermediary sets, before the set $S_{r}(n)$ can finally be determined. Thus how are we to store, search and retrieve the solutions contained in these intermediary sets?  As shall be seen later, this question will be solved by the use of a commonly occurring data structure known as a Binary Search Tree, which can efficiently search and retrieve stored data in the form of nodes within a tree structure. In our case, the intermediary sets and the set $S_{r}(n)$, will form the nodes of an evolving Binary Search Tree as pictured in Figures 1, 2, and 3.\\
\\
Although a complexity analysis will not be performed here, one can at least conclude from the recursive procedure described above that the algorithm must terminate, for each input $n > 2$. We now briefly outline the structure of the remaining paper. In Section 2, we begin by proving a number of preliminary results leading to the structure of the solution set $S(n)$, and then introduce the main features of the Binary Search Tree. Within Section 3 our task will be to establish the recursive procedure, which will then be formulated into the pseudo code of Algorithm 3.1. This will be followed by an examination of the use of Binary Search Trees that will be used to implement the algorithm into a working program, which the interested reader can access at \cite{program1}. 
 Finally in Section 4, we investigate a theoretical application of Algorithm 3.1 to a well known conjecture connected with the Diophantine equation in (\ref{eq:2ex}). This conjecture asserts that the only integers $n>2$, for which the $n$-tuple $(2,n;n-2)$ is a unique solution to (\ref{eq:2ex}), are those contained in the set $E=\{ 2,3,4,6,24,114,174,444\}$, the so-called set of exceptional values. As reported in \cite{ecker}, despite extensive computer searches of all positive integers less than $10^{10}$, no new elements of $E$ have been revealed. Although we will not resolve the conjecture here, we shall by using an argument based on the above recursive procedure, prove that if $n>2$ is an element of $E$, then $n-1$ must be a Sophie Germain prime, that is both $n-1$ and $2(n-1)+1$ are prime. Equipped with this result together with the scarcity of the Sophie Germain primes, we can see that Algorithm 3.1 provides a basis for a more refined computer search for possible extra elements of $E$.

\section{Preliminaries}
In this section we shall first establish some preliminary technical lemmas that will 
 be needed in Section 3. 
 In addition, a brief overview of the data structure
 known as a Binary Search Tree, which will be used for the implementation of the
 algorithm, shall also be given. To begin we note that the first of the required
 lemmas was proved in \cite{ecker},\cite{kurland}, but we present here an alternate proof based
 on a divisibility argument, which later will form the basis for the construction
 of the algorithm. The following result states that in any solution to the ESP-Problem
 in $n>2$ variables, there must be at least one unit component.
\begin{lem}
The ESP-Problem in two variables has exactly one solution, namely
 $\{ (x_{1},x_{2})\in\mathbb{N}\times\mathbb{N}:x_{1}x_{2}=x_{1}+x_{2}\} =\{ (2,2)\}$.
 While if $(x_{1},x_{2},\ldots ,x_{n})$ is a solution to the ESP-Problem in $n\geq 3$
 variables, then there must exist at least one $i\in\{ 1,2,\ldots ,n\}$ such that $x_{i}=1$.
\end{lem}
\begin{proof}
Suppose $x_{1}x_{2}=x_{1}+x_{2}$, for some positive integers $x_{1},x_{2}$. Then
 as $x_{1}|(x_{1}+x_{2})$, observe $x_{1}|x_{2}$ and so $x_{1}\leq x_{2}$. Similarly
 $x_{2}\leq x_{1}$. Consequently $x_{1}=x_{2}$ and so $x_{1}^{2}=2x_{1}$ which yields
 that either $x_{1}=2$ or $x_{1}=0$, but as $x_{1}$ is a positive integer we conclude
 $x_{1}=x_{2}=2$. Next we show that apart from the two-variable case, the ESP-Problem
 in $n\geq 3$ variables cannot have solutions $(x_{1},x_{2},\ldots ,x_{n})$ in which
 $x_{i}\geq 2$ for all $1\leq i\leq n$. This can be proved using the following argument:
 Assuming $x_{i}\geq 2$ for all $1\leq i\leq n$ and recalling that 
 $x_{i}\leq x_{i+1}$, observe that $x_{1}x_{2}\cdots x_{n}\geq 2^{n-1}x_{n}$
 while $n x_{n} \geq x_{1}+x_{2}+\ldots +x_{n}$, but as $2^{n-1}>n$ for $n\geq 3$ we deduce
 that $x_{1}x_{2}\cdots x_{n}> x_{1}+x_{2}+\ldots +x_{n}$, and so the $n$-tuple $(x_{1},x_{2},\ldots,x_{n})$ 
 cannot be a solution of (\ref{eq:2ex}).
\end{proof}\\
In view of the previous result we introduce, for notational convenience, the following
 definition for classifying solutions of the ESP-Problem in $n$ variables, according to the number of non-unit
 components present within the $n$-tuple.
\begin{defy}
For given integers $n\geq r\geq 2$, let $(x_{1},x_{2},\ldots ,x_{r} ;n-r)$ denote an
 $n$-tuple having $r$ non-unit components and satisfying equation (\ref{eq:2ex}), and set
\[
S_{r}(n)=\{ (x_{1},x_{2},\ldots ,x_{r} ; n-r)\in \mathbb{N}^{n}: \prod_{j=1}^{r} x_{j} =\sum_{j=1}^{r} x_{j} + n-r\}\mbox{ .}
\]
\end{defy}
Clearly from Definition 2.1 we observe that $S_{2}(n)\not=\emptyset$, as the n-tuple  
 $(2,n;n-2)$, is an element of $S_{2}(n)$ for
 all $n\geq 2$. In the next result a characterization for the set $S_{2}(n)$ will
 be given in terms of the divisors of $n-1$.
\begin{lem}
For an integer $n\geq 2$, the set $S_{2}(n)$ has the following explicit form:
\begin{equation}
S_{2}(n)=\left\{ \left( d+1,\frac{n-1}{d}+1;n-2\right) :d| (n-1), d\leq\sqrt{n-1}\right\}.\label{eq:1ex}
\end{equation}
\end{lem}
\begin{proof}
Suppose $(x_{1},x_{2};n-2)\in S_{2}(n)$. Then upon rearrangement of
$x_{1}+x_{2}+n-2=x_{1}x_{2}$, we find $n-1=(x_{1}-1)(x_{2}-1)$. Assuming
 $x_{1}\leq x_{2}$, observe that for a divisor $d$ of $n-1$ with $d\leq \sqrt{n-1}$,
 we may set $x_{1}-1=d$ and $x_{2}-1=(n-1)/d$ and so $S_{2}(n)$ is of the form as
 stated in (\ref{eq:1ex}).
\end{proof}\\
A variation of the following result was proved in \cite{ecker},\cite{kurland}, and gives a lower
 bound on $r$ that insures $S_{r}(n)=\emptyset$. We present here for completeness,
 the proof which uses the fact that any solution $(x_{1},x_{2},\ldots ,x_{r};n-r)\in S_{r}(n)$
 has the property that the common equal-sum-product value is at most $2n$, with equality
 holding if and only if when $r=2$ and  $(x_{1},x_{2},\ldots ,x_{r};n-r)=(2,n;n-2)$.
 This result incidentally appeared in the form of a problem for the Polish Mathematical
 Olympiad in 1990 (see \cite{kurland}).
\begin{lem}
For an integer $n\geq 2$, if $r>\lfloor \log_{2}(n)\rfloor +1$, then $S_{r}(n)=\emptyset$.
\end{lem}
\begin{proof}
Suppose $S_{r}(n)\not=\emptyset$, for some $n\geq 2$. Now in \cite{ecker}, it was shown
 that any $(x_{1},x_{2},\ldots ,x_{r};n-r)\in S_{r}(n)$ has the property that the
 common equal-sum-product value is bounded above by $2n$. Consequently as
 $2^{r}\leq \prod_{l=1}^{r}x_{l}=x_{1}+\cdots x_{r}+n-r\leq 2n$ we deduce that
 $2^{r-1}\leq n$, from which one finds $r\leq \log_{2}(n)+1$. As $r$ is an integer
 the previous inequality implies $r\leq \lfloor\log_{2}(n)\rfloor +1$.
\end{proof}
\\
For any integer $n>2$, as the number of non-unit components in a solution
 $(x_{1},x_{2},\ldots ,x_{r};n-r)\in S(n)$ is equal to $r\geq 2$, we are guaranteed
 that at least two integers say $x',x''\in \{x_{1},x_{2},\ldots ,x_{r}\}$ are such
 that $x''\geq x'\geq 2$. Now if $x''$ is the largest non-unit component of a solution
 $(x_{1},x_{2},\ldots ,x_{r};n-r)\in S(n)$, then from the above upper bound for
 the common equal-sum-product value, we find $2x''\leq \prod_{i=1}^{r}x_{i}\leq 2n$,
 that is $x''\leq n$. Consequently each of the non-unit components in a solution
 $(x_{1},x_{2},\ldots , x_{r};n-r)\in S_{r}(n)$ can only assume the $n-1$ integer
 values in the set $\{ 2,\ldots ,n\}$, and
 so 
 $|S_{r}(n)| \leq (n-1)^{r}$.
 Thus one deduces from Lemma 2.3 that the reduced search space for the
 ESP-Problem in $n$ variables is
 $(n-1)^{2}+(n-1)^{3}+\cdots +(n-1)^{\lfloor \log_{2}(n)\rfloor +1}=O((n-1)^{\lfloor \log_{2}(n)\rfloor+1})$.
 In view of this bound, we can see that an exhaustive search for solutions to the
 ESP-Problem is impractical for large $n$, and hence the need for an algorithmic solution
 to this problem. \\
\\
As shall be seen in Section 3, one of the main operations performed by the algorithm
 to recursively generate the solution sets $S_{r}(n)$, will be to search previously
 constructed and stored sets of the form $S_{r-1}(\cdot )$, for the eventual purpose
 of applying a divisibility test. To make this searching operation practicable, we
 first impose on the solution sets $S_{r}(n)$ a partial order defined as
 $S_{r_{1}}(n_{1})<S_{r_{2}}(n_{2})$ when $r_{1}<r_{2}$ in the case $n_{1}=n_{2}$, while
 in the case $n_{1}\not= n_{2}$ then $S_{r_{1}}(n_{1})<S_{r_{2}}(n_{2})$ when
 $n_{1}<n_{2}$. This partial ordering can best be summarized using the Iverson bracket
 notation (see \cite[p.24]{graham}) as follows
\begin{equation}
[S_{r_{1}}(n_{1})<S_{r_{2}}(n_{2})]=[n_{1}=n_{2}][r_{1}<r_{2}] + [n_{1}<n_{2}]\mbox{ ,}\label{parord}
\end{equation}
where the square bracket $[P]$ evaluates to 1 if the statement $P$ is true and 0
 otherwise. Note from the definition of the partial ordering,
 $S_{r_{1}}(n_{1})=S_{r_{2}}(n_{2})$ if and only if $r_{1}=r_{2}$ and $n_{1}=n_{2}$.
 In addition to the partial order, the second device we shall employ to facilitate
 the searching operation of the algorithm, will be the storage of the solution sets
 $S_{r}(n)$ as nodes in a Binary Search Tree. A Binary Search Tree, denoted $T$, is
 an abstract data structure used for the storage, retrieval and deletion of nodes
 representing the elements of a finite set $S$, on which a partial order is defined.
 More precisely a Binary Search Tree can be defined via the following definition:
\begin{defy}
A Binary Search Tree (BST) is a data structure $T$ for a finite partially ordered
 set $S$, and is defined as follows: If $S=\emptyset$ then $T$ is a null tree, otherwise
 $T$ consists of a node storing some element $x\in S$ and two Binary Search Trees such that:\\
\\
(1) $Left(T)$ stores elements of $S$ which are less than $x$.\\
\\
(2) $Right(T)$ stores elements of $S$ which are greater than $x$.\\
\\
The node storing $x$ is called the root node of the tree.
\end{defy}
To illustrate the concept of a Binary Search Tree and how it will be used for the
 storage of the solution sets $S_{r}(n)$ during the execution of the algorithm, consider
 the following set $S=\{ S_{2}(2),S_{3}(4),S_{3}(5),S_{3}(6),S_{2}(7),S_{4}(15)\}$,
 whose elements have been placed in ascending order, as specified by the partial
 ordering in (\ref{parord}). If for example we select the element $S_{3}(5)\in S$ as the root node of
 the BST, then $Left(T)$ would be a BST consisting of nodes contained in the subset
 $\{ S_{2}(2),S_{3}(4)\}$, while $Right(T)$ would also be a BST consisting of nodes
 contained in the remaining subset $\{ S_{3}(6),S_{2}(7),S_{4}(15)\}$. By further selecting
 $S_{2}(2)$ and $S_{2}(7)$ as the root nodes of $Left(T)$ and $Right(T)$ respectively,
 then the elements of $S$ can now be stored within the BST pictured in Figure 1. \\
\\
Defining the height of a BST as the total number of levels below the root node, observe
 that the tree structure in Figure 1 has a height of $2$. Clearly from Definition
 2.2 one can see there are many possible BST that can be constructed to store the
 elements, as nodes, of a finite partially ordered set $S$. By using standard algorithms
 for the transformation of BST, (see \cite[pp. 458-481]{knuth}), a BST can have it's height reduced,
 while maintaining the original partial ordering of the set $S$. A BST is said to
 be balanced when it's height is reduced to a minimum. One major advantage of a
 balanced BST having $n$ nodes, is that the operations of searching and insertion
 of a node can be performed in $O(\log n)$ time in the worst case, compared with
 a linear list where the same operations take $O(n)$ time in the worst case. This
 efficiency is the reason for our use of such a data structure in the implementation
 of the algorithm in Section 3.
\\
\begin{figure}[htb]
\centering
{\scalefont{\tikzpicturefont}
\begin{tikzpicture}[]
\node [circle,draw] (z){$ S_{3}(5) $}
  child 
  {
    node [circle,draw] (a) {$ S_{2}(2) $}
    child 
    { 
      [draw=white]
      node [draw=white] (b) {$  $}
    }
    child 
    { 
      node [circle,draw] (c) { $ S_{3}(4) $ }
    }
  }
  child 
  {
    node [circle,draw] (d) { $ S_{2}(7) $ } 
    child 
    {
      node [circle,draw] (e) { $ S_{3}(6) $ }
    }
    child 
    {
      node [circle,draw] (f) { $S_{4}(15)$ }
    }
  };
\end{tikzpicture}
}
\caption*{Figure 1}
%\label{Subfloat label}
\end{figure}
\section{Algorithm Construction and Implementation}
%{\bf Towards an algorithm.}\\
We now come to the heart of the paper where an algorithm will be constructed
 which recursively generates all solutions
 of the ESP-Problem in $n$ variables, and where an examination of the algorithms
 implementation shall also be given. To begin, in view of Lemma 2.3, we can readily
 deduce that the complete solution set of the ESP-Problem in $n$ variables is given by
\[
S(n)=\bigcup_{r=2}^{m+1}S_{r}(n)\mbox{ ,}
\]
where $m=\lfloor \log_{2} n\rfloor$. Applying (\ref{eq:1ex}) of Lemma 2.2 it is a
 theoretically straightforward procedure to construct the set $S_{2}(n)$ for any
 integer $n\geq 2$. However in what follows we will first describe a general process
 by which each set $S_{r}(n)$, for $r=3,\ldots ,m+1$ and $n\geq 3$, can be recursively
 generated from the elements of the sets $S_{r-1}(r+j)$, for values of $j$ to be
 specified shortly. This process, established in Step 1, will form the basis for
 the construction of Algorithm 3.1 described in Step 2 (Part A), while in Step 2
 (Part B) the implementation of Algorithm 3.1 using balanced BST will be examined,
 and shall lead to a working program in the form of a javascript web page, (see \cite{webpage}),
 for the generation of all solutions to the ESP-Problem in $n>2$ variables.\\
\\
{\bf Step 1: Recursive generation of $S_{r}(n)$.}\\
Suppose $S_{r}(n)\not=\emptyset$ for $r\geq 3$, then there must exist $r$ integers
 greater than unity and such that $x_{1}x_{2}\ldots x_{r}=x_{1}+x_{2}+\ldots +x_{r}+n-r$.
 Solving for $x_{r}$ we find that
\begin{equation}
x_{r}=\frac{x_{1}+x_{2}+\cdots +x_{r-1}+n-r}{x_{1}x_{2}\ldots x_{r-1}-1}\in\mathbb{N}\mbox{ ,}\label{eq:1}
\end{equation}
and so $(x_{1}x_{2}\cdots x_{r-1}-1)|(x_{1}+x_{2}+\ldots + x_{r-1}+n-r)$. Now as
 $x_{r}\geq 2$, we can conclude that $(x_{1}x_{2}\cdots x_{r-1}-1)<x_{1}+x_{2}+\cdots +x_{r-1}+n-r$
 and so, there must exist a $k\in\mathbb{N}$ such that
 $x_{1}x_{2}\cdots x_{r-1}-1=(x_{1}+x_{2}+\cdots +x_{r-1}+n-r)-k$, or equivalently
 after setting ${\bf j=n-r-k}$
\begin{equation}
x_{1}x_{2}\cdots x_{r-1}=x_{1}+x_{2}+\cdots +x_{r-1}+(j+1)\mbox{ ,}\label{eq:2}
\end{equation}
with $j+1\geq 0$ by Lemma 2.1. Thus
 $(x_{1},x_{2},\ldots ,x_{r-1};j+1)\in S_{r-1}(r+j)$. Conversely assume
 $(x_{1},x_{2},\ldots ,x_{r-1};j+1)\in S_{r-1}(r+j)$, with ${\bf r+j<n}$ and is
 such that $x_{r}:=(\sum_{l=1}^{r-1}x_{l}+n-r)/(\prod_{l=1}^{r-1}x_{l}-1)\in\mathbb{N}$,
 then by definition of $x_{r}$ we have $(x_{1},x_{2},\ldots ,x_{r-1},x_{r};n-r)\in S_{r}(n)$,
 noting here that such an $x_{r}\geq 2$, since by assumption
 $\prod_{l=1}^{r-1}x_{l}=\sum_{l=1}^{r-1}x_{l}+j+1$ we have
\begin{eqnarray*}
x_{r}:=\frac{\sum_{l=1}^{r-1}x_{l}+n-r}{\prod_{l=1}^{r-1}x_{l}-1} = \frac{\sum_{l=1}^{r-1}x_{l}+n-r}{\sum_{l=1}^{r-1}x_{l}+j} & = & \frac{(\sum_{l=1}^{r-1}x_{l}+j)+n-r-j}{\sum_{l=1}^{r-1}x_{l}+j} \nonumber\\
& = & 1+\frac{n-r-j}{\sum_{l=1}^{r-1}x_{l}+j}>1\mbox{ .}\\
\end{eqnarray*}
 Thus to generate $S_{r}(n)$, it is necessary and sufficient to find solutions
 $(x_{1},x_{2},\ldots ,x_{r-1};j+1)\in S_{r-1}(r+j)$, with $j+1 \geq 0$, satisfying the divisibility
 condition in (\ref{eq:1}), and construct $S_{r}(n)$ as the set containing elements
 of the form $(x_{1},x_{2}\ldots ,x_{r-1},x_{r};n-r)$, with $x_{r}$ defined as in
 (\ref{eq:1}). We next give upper and lower bound on $j$ necessary for both
 $S_{r-1}(r+j)\not=\emptyset$ and have elements that satisfy the divisibility condition
 in (\ref{eq:1}). Recall from Lemma 2.3 that for $S_{r}(n)\not=\emptyset$ necessarily
 $r\leq\lfloor\log_{2}(n)\rfloor +1$, consequently if $S_{r-1}(r+j)\not=\emptyset$,
 then one finds $r-1\leq\lfloor\log_{2}(r+j)\rfloor +1\leq\log_{2}(r+j)+1$, which upon
 rearrangement yields that
\begin{equation}
2^{r-2}-r\leq j\mbox{ .}\label{eq:3}
\end{equation}
Recalling from above that if $(x_{1},x_{2},\dots ,x_{r-1};j+1)\in S_{r-1}(r+j)$ we
 again have
\begin{equation}
\frac{\sum_{l=1}^{r-1}x_{l}+n-r}{\prod_{l=1}^{r-1}x_{l}-1} =1+\frac{n-r-j}{\sum_{l=1}^{r-1}x_{l}+j}\mbox{ .}\label{eq:4}
\end{equation}
Now the right-hand side of (\ref{eq:4}) will not be an integer if $n-r-j<\sum_{l=1}^{r-1}x_{l}+j$,
 but as $x_{l}\geq 2$, further observe that the previous inequality will be satisfied if and only if
$n-r-j<2(r-1)+j$, that is when
\begin{equation}
j>\frac{n-3r+2}{2}\Leftrightarrow j>\left\lfloor \frac{n-3r+2}{2}\right\rfloor\mbox{ ,}\label{eq:5}
\end{equation}
as $j\in\mathbb{N}$. Thus for $S_{r-1}(r+j)\not=\emptyset$ and to have elements
 satisfying the divisibility condition of (\ref{eq:1}), we must necessarily have
\begin{equation}
2^{r-2}-r\leq j\leq \left\lfloor\frac{n-3r+2}{2}\right\rfloor\mbox{ ,}\label{eq:6}
\end{equation}
assuming that  $2^{r-2}-r\leq\left\lfloor\frac{n-3r+2}{2}\right\rfloor$.\\
\\
{\bf NB:}\,\, 
  If $j \geq 2^{r-2}-r \geq \lfloor \frac{n-3r+2}{2} \rfloor$ then 
 no element in $S_{r-1}(r+j)$ can satisfy the divisibility condition in (\ref{eq:1}) and so $S_{r}(n) = \emptyset $. \\
\\
Thus excluding the possible case of
 $j\geq 2^{r-2}-r>\left\lfloor\frac{n-3r+2}{2}\right\rfloor$ in which
 $S_{r}(n)=\emptyset$, we can conclude that to construct each $S_{r}(n)$,
 for $r=3,\ldots ,m+1$ and $n\geq 3$, it suffices to examine the elements
 $(x_{1},x_{2},\ldots, x_{r-1};j+1)\in S_{r-1}(r+j)$, for the values of $j$ in
 (\ref{eq:6}) which satisfy the divisibility condition
\[
w=1+\frac{n-r-j}{\sum_{l=1}^{r-1}x_{l}+j}\in\mathbb{N}\mbox{ ,}
\]
and then construct $S_{r}(n)$ from the set of elements of the form
 $(x_{1},x_{2},\ldots ,x_{r-1},w;n-r)$.\\
\\
{\bf Step 2 Part A: Construction of Algorithm}\\
We now formalize the recursive generation of the sets $S_{r}(n)$ in Step 1, into
 the pseudo code of Algorithm 3.1. To assemble the required algorithm recall that
 the solution set of the ESP-Problem in $n$ variables is $S(n)=\bigcup_{r=2}^{m+1}S_{r}(n)$, 
 in which
 the set $S_{2}(n)$ is explicitly given in (\ref{eq:1ex}), while from Step 1 each
 set $S_{r}(n)$ for $r=3,\ldots ,m+1$ is of the form
\begin{equation}
\bigcup_{j=2^{r-2}-r}^{\lfloor\frac{n-3r+2}{2}\rfloor}\{ (x_{1},x_{2},\ldots ,x_{r-1},w;n-r):(x_{1},x_{2},\ldots ,x_{r-1};j+1)\in S_{r-1}(r+j),w=1+\frac{n-r-j}{\sum_{l=1}^{r-1}x_{l}+j}\in\mathbb{N}\}\mbox{ ,}\label{eq:7}
\end{equation}
provided $\lfloor\frac{n-3r+2}{2}\rfloor\geq 2^{r-2}-r$, with $S_{r}(n)=\emptyset$
 otherwise. Now to affect the construction of each set $S_{r}(n)$ via (\ref{eq:7}),
 we define a recursive procedure denoted $calc\_shell(n,r)$ in a typical programming
 language such as C++. Introducing a generic data structure, denoted $T$, and initially
 assigned $T:=\emptyset$, for storing and accessing all previously constructed sets
 $S_{r-1}(j+r)$ for $2^{r-2}-r\leq j\leq \lfloor\frac{n-3r+2}{2}\rfloor$, the procedure
 $calc\_shell(n,r)$ will first test if the set $S_{r}(n)$ is already contained in
 $T$ and exist if ($S_{r}(n)\in T$) is true.\\
 \\
 When ($S_{r}(n)\in T$) is false, then after initially assigning $S_{r}(n):=\emptyset$,
 if $r=2$ construct $S_{r}(n)$ using the explicit formulation in (\ref{eq:1ex}) and
 insert $S_{r}(n)$ into $T$ via the assignment $T:=T\cup S_{r}(n)$ and exit. However
 if $r>2$, then the procedure $calc\_shell(n,r)$, as dictated by (\ref{eq:7}), enters
 into a {\scriptsize FOR} Loop indexed by $j=\lfloor\frac{n-3r+2}{2}\rfloor...2^{r-2}-r$, which either
 terminates if $2^{r-2}-r>\lfloor\frac{n-3r+2}{2}\rfloor$ resulting in the assignment
 $S_{r}(n):=\emptyset$ or proceeds for each $j$ to access the elements of the sets
 $S_{r-1}(r+j)$ from $calc\_shell(j+r,r-1)$ if non-empty, and tests the divisibility
 condition $w=1+\frac{n-r-j}{\sum_{l=1}^{r-1}x_{l}+j}\in\mathbb{N}$, which if true inserts
 the solution $(x_{1},x_{2},\ldots ,x_{r-1},w;n-r)$ into the set $S_{r}(n)$ 
 or
 otherwise proceeds to the next value of the index $j$. Once this {\scriptsize FOR} Loop is completed
 the resulting set $S_{r}(n)$ is inserted into the data structure $T$ via the assignment
 $T:=T\cup S_{r}(n)$. Finally to construct the complete solution set $S(n)$ of the
 ESP-Problem, each recursive procedure $calc\_shell(n,i)$ is executed within the procedure
 $cal\_solution(n)$, which consists of a {\scriptsize FOR} Loop indexed by $i=m+1...2$. The above
 algorithm assembly is summarized in pseudo code as follows:\\
\\
{\bf Algorithm 3.1}\\
\\
{\bf Initialization:} $ T := \emptyset $ \\
 \\
Construct the set $S_{r}(k)$\\
$ procedure \hspace*{2mm} calc\_shell(k,r)$\\
  \mbox{ }\hspace{1cm}If $S_{r}(k) \in T $ then\\
  \mbox{ }\hspace{1cm} exit
  \\
  \mbox{ }\hspace{1cm}$S_{r}(k):=\emptyset$ \\
  \mbox{ }\hspace{1cm}If $r = 2 $ then
 $S_{2}(k):=\left\{ \left(\frac{k-1}{d}+1,d+1;k-2\right) :d| (k-1), d\leq\sqrt{k-1}\right\}$ \\
    \mbox{ }\hspace{2cm}$ T := T \cup  S_{2}(k) $ \\
    \mbox{ }\hspace{2cm}exit \\
  \mbox{ }\hspace{1cm}For $ j = \lfloor\frac{k-3r+2}{2}\rfloor \ldots 2^{r-2}-r $ \\
    \mbox{ }\hspace{2cm}If $2^{r-2}-r>\lfloor\frac{k-3r+2}{2}\rfloor$
 then $S_{r}(k):=\emptyset\,\,\mbox{and}\,\, T:=T\cup S_{r}(k)$\\
      \mbox{ }\hspace{2cm}exit \\
    \mbox{ }\hspace{2cm} call $calc\_shell(j+r,r-1) $ \\
    \mbox{ }\hspace{2cm} For each element
 $(x_{1},x_{2}\ldots ,x_{r-1};j+1)\in S_{r-1}(j+r)$ such that \\
\[
w=1+\frac{k-r-j}{\sum_{l=1}^{r-1}x_{l} +j}\in\mathbb{N}
\]
      \mbox{ }\hspace{3cm} then $S_{r}(k):=S_{r}(k)\cup\{ (x_{1},x_{2},\ldots ,x_{r-1},w;k-r)\}$\\
  \mbox{ }\hspace{1cm}$ T := T \cup  S_{r}(k) $ \\
 \\
Construct the solution set $S(n)$\\
$ procedure \hspace*{2mm} calc\_solution(n)$\\
  \mbox{ }\hspace{1cm}For $ i= \lfloor\log_{2}(n)\rfloor +1 \ldots 2 $ \\
    \mbox{ }\hspace{2cm} $ calc\_shell(n,i) $ \\
\mbox{}\\
%\vspace{1cm}
{\bf Step 2 Part B: Implementation of Algorithm}\\
To illustrate both the function of Algorithm 3.1 and 
 the need for a balanced
 Binary Search Trees (BST) in it's implementation, we examine first the construction
 of the solution set for the ESP-Problem in the case of $n=15$ as follows.\\
\\
{\bf Example 3.1:} As $2^{3}<15<2^{4}$ we have $m=\lfloor\log_{2}15\rfloor =3$ and
 so $S(15)=\cup_{r=2}^{4}S_{r}(15)$. After initializing $T:=\emptyset$, Algorithm
 3.1 executes each procedure $calc\_shell(15,i)$ in descending order of $i=4,3,2$,
 to construct $S_{i}(15)$, initially assigned as empty. Beginning with $calc\_shell(15,4)$,
 this procedure must first call on each $calc\_shell(j+4,3)$ in descending order of
 $j=2,1,0$, for the construction of $S_{3}(6)$,$S_{3}(5)$,$S_{3}(4)$ respectively,
 before invoking the divisibility test
\begin{equation}
(x_{1},x_{2},x_{3};j+1)\in S_{3}(j+4)\,\,\mbox{is such that}\,\, w=1+\frac{11-j}{\sum_{l=1}^{3}x_{l}+j}\in\mathbb{N}\mbox{ ,}\label{eq:8}
\end{equation}
which if true results in the insertion $(x_{1},x_{2},x_{3},w;11)\in S_{4}(15)$,
 or else leaves $S_{4}(15):=\emptyset$. At this stage we similarly find that the
 execution of each $calc\_shell(k,3)$ for $k=6,5$ must both call on
 $calc\_shell(j+3,2)$ for $j=-1$ to produce, via (\ref{eq:1ex}), the set
 $S_{2}(2)=\{ (2,2;0)\}$, before invoking the divisibility test
\begin{equation}
(2,2;0)\in S_{2}(2)\,\,\mbox{ is such that}\,\, w=1+\frac{k-3+1}{4-1}\in\mathbb{N}\mbox{ ,}\label{eq:9}
\end{equation}
which if true results in the insertion $(2,2,w;k-3)\in S_{3}(k)$ or else leaves
 $S_{3}(k):=\emptyset$. Consequently upon substituting $k=6,5$ into (\ref{eq:9})
 we deduce that $S_{3}(6)=\emptyset$ while $S_{3}(5)=\{ (2,2,2;2)\}$. However the
 execution of $calc\_shell(4,3)$ must result in $S_{3}(4):=\emptyset$ as
 $2^{r-2}-r>\lfloor\frac{k-3r+2}{2}\rfloor$ for $(k,r)=(4,3)$. Thus by substituting
 the only solution $(2,2,2;2)=(x_{1},x_{2},x_{3};j+1)\in S_{3}(j+4)$, corresponding
 to $j=1$, into (\ref{eq:8}), we find $w=1+\frac{10}{7}\not\in\mathbb{N}$ and so
 $S_{4}(15):=\emptyset$. Thus $T$ remains empty.\\
\\
In like manner the execution of $calc\_shell(15,3)$, to construct $S_{3}(15)$,
 will call on $calc\_shell(j+3,2)$ in descending order of $j=4,3,2,1,0,-1$, and 
produces again via (\ref{eq:1ex}) the sets
 $S_{2}(7)=\{(7,2;5),(4,3;5)\},S_{2}(6)=\{(6,2;4)\},S_{2}(5)=\{(5,2;3),(3,3;3)\},S_{2}(4)=\{(4,2;2)\},S_{2}(3)=\{(3,2;1)\},S_{2}(2)=\{ (2,2;0)\}$,
 which after substituting into the divisibility test
\[
(x_{1},x_{2};j+1)\in S_{2}(j+3)\,\,\mbox{is such that}\,\, w=1+\frac{12-j}{\sum_{l=1}^{2}x_{l}-j}\in\mathbb{N}\mbox{ ,}
\]
yields that $S_{3}(15)=\emptyset$, as for each corresponding value of $j$ we find
 $w\not\in\mathbb{N}$, and so $T$ remains empty. Finally the execution of $calc\_shell(15,2)$
 produces via (\ref{eq:1ex}) the only non-empty set, namely
$S_{2}(15)=\{ (15,2;13), (8,3;13)\}$ for insertion into $T$. Thus we conclude
 $S(15)=\{ (15,2;13), (8,3;13)\}$.\\
\\
From Example 3.1 we can see that even for a small value of $n=15$, prior to the
 construction of each set $S_{r}(n)$, for $r=\lfloor\log_{2}n\rfloor +1\ldots 3$,
 the procedure $calc\_shell(n,r)$ must first call on itself a number of times in
 descending procedural calls of the form $calc\_shell(j+r,r-1)$, for
 $j=\lfloor\frac{n-3r+2}{2}\rfloor \ldots 2^{r-2}-r$, in order to construct the
 sets $S_{r-1}(r+j)$. However before these later sets can be constructed, a similar
 process of descending procedural calls must be repeatedly applied until we reach
 the construction of the ``base'' sets namely $S_{2}(\cdot)$ 
 given explicitly in (\ref{eq:1ex}). Once these base sets have been determined,
 then by applying the respective divisibility tests in the reverse order to the
 sequence of procedural calls described above, one can finally construct each set
 $S_{r}(n)$, and obtain the complete solution set $S(n)$. \\
\\
Clearly, the first problem we must overcome to achieve a practicable implementation
 of Algorithm 3.1, is to efficiently store and retrieve solutions within the respective
 sequence of intermediary sets that are needed for the recursive generation of each
 set $S_{r}(n)$, for $r=\lfloor\log_{2}n\rfloor +1\ldots 3$. Moreover, as illustrated
 by Example 3.1, we should avoid the unnecessary reconstruction of existing intermediary
 sets, which may arise when constructing the sets $S_{r}(n)$ for lower values of $r$.
 Finally one will also require a practical method for determining the base sets
 $S_{2}(\cdot )$ using (\ref{eq:1ex}). This problem will be addressed last, 
 and in the interim we shall
 assume that one can readily construct the set $S_{2}(n)$ for any integer $n>2$.\\
\\
Beginning with the problem of storage and retrieval, recall from Section 2 that a
 finite collection of solution sets such as
 $S=\{ S_{2}(2),S_{3}(4),S_{3}(5),S_{3}(6),S_{2}(7),S_{4}(15)\}$, can be inserted
 as nodes in a BST which preserves the partial
 ordering in (\ref{parord}). When applying Algorithm 3.1 after tracking the execution
 of all descending procedural calls required in the construction of each set $S_{r}(n)$,
 we first determine a specific collection of base sets $S_{2}(\cdot)$ which are
 constructed in descending order with respect to (\ref{parord}). Once determined
 these base sets can then be inserted into an evolving BST, in which their placement
 into the left or right subtree will depend upon their relative ordering via (\ref{parord}) with respect
 to the current root node. Upon re-balancing, the process of retrieving the required
 solutions for applying the respective divisibility tests can then be executed, and
 the resulting intermediary sets, whether empty or not, can be progressively inserted
 into an evolving BST, culminating with the insertion of each set $S_{r}(n)$. Once all
 the sets $S_{r}(n)$ have been inserted into the final balanced BST, the complete
 solution set $S(n)$ can then be formed by retrieval of the elements from 
 those sets $S_{r}(n)$ which are non-empty.\\
\\
One advantage in using a balanced BST is that we can search through $m$ nodes in
 $O(\log (m))$ time, thus making efficient the process of searching and retrieval.
 This efficiency can further be exploited to effectively search for pre-existing
 intermediary sets, thus avoiding unnecessary reconstruction. A secondary advantage
 in using BST is that these data structures are readily implementable in a programming
 language such as C++, where the operations of storing, searching and retrieval of
 nodes can be performed using existing Standard Template Libraries supported within
 C++. We now illustrate the process outlined above by displaying a subsequence
 of evolving balanced BST leading to the construction of $S(15)$.\\
\\
From Example 3.1 recall the algorithm began with the construction of $S_{4}(15)$ 
by first executing three descending procedural calls, to construct the intermediary sets 
 $S_{3}(6)$,$S_{3}(5)$,$S_{3}(4)$, which in turn could not be determined until 
 the base set $S_{2}(2)$ was constructed. Once this is achieved, the sequence of 
 insertions into a BST can then begin with $S_{2}(2)$ chosen as the initial root 
 node followed by the insertion of the set $S_{3}(6)$ onto the right of $S_{2}(2)$. 
 Next after inserting $S_{3}(5)$ to the right of $S_{3}(6)$ and then re-balancing, 
 the root node of the new BST then becomes $S_{3}(5)$. After insertion of the remaining 
 set $S_{3}(4)$ to the right of $S_{2}(2)$, the BST can now be searched for solutions 
 used in the divisibility test to construct $S_{4}(15)$, which is then placed to 
 the right of $S_{3}(6)$ in the BST. This initial sequence of evolving BST is illustrated 
 below in Figure 2.\\
\begin{figure}[htb]
\centering
% d001.tex
{\scalefont{\tikzpicturefont}
\begin{tikzpicture}[] 
\node [circle,draw] (z){$ S_{2}(2) $}
  child 
  {
    [draw=white]
    node [draw=white] (a) {$ $}
    child 
    { 
      [draw=white]
      node [draw=white] (b) {$  $}
    }
  };
\end{tikzpicture}
}
% d002.tex
{\scalefont{\tikzpicturefont}
\begin{tikzpicture}[] 
\node [circle,draw] (z){$ S_{2}(2) $}
  child 
  {
    [draw=white]
    node [draw=white] (a) {$ $}
  }
  child 
  {
    node [circle,draw] (b) {$ S_{3}(6) $}
  };
\end{tikzpicture}
}
% d003.tex
{\scalefont{\tikzpicturefont}
\begin{tikzpicture}[] 
\node [circle,draw] (z){$ S_{3}(5) $}
  child 
  {
    node [circle,draw] (a) {$ S_{2}(2) $}
  }
  child 
  {
    node [circle,draw] (b) {$ S_{3}(6) $}
  };
\end{tikzpicture}
}
% d004.tex
{\scalefont{\tikzpicturefont}
\begin{tikzpicture}[] 
\node [circle,draw] (z){$ S_{3}(5) $}
  child 
  {
    node [circle,draw] (a) {$ S_{2}(2) $}
    child 
    { 
      [draw=white]
      node [draw=white] (c) {$  $}
    }
    child 
    { 
      node [circle,draw] (d) { $ S_{3}(4) $ }
    }
  }
  child 
  {
    node [circle,draw] (b) {$ S_{3}(6) $}
  };
\end{tikzpicture}
}
% d005.tex
{\scalefont{\tikzpicturefont}
\begin{tikzpicture}[]
\node [circle,draw] (z){$ S_{3}(5) $}
  child 
  {
    node [circle,draw] (a) {$ S_{2}(2) $}
    child 
    { 
      [draw=white]
      node [draw=white] (b) {$  $}
    }
    child 
    { 
      node [circle,draw] (g) { $ S_{3}(4) $ }
    }
  }
  child 
  {
    node [circle,draw] (j) { $ S_{3}(6) $ } 
    child 
    {
      [draw=white]
      node [draw=white] (k) {  }
    }
    child 
    {
      node [circle,draw] (l) { $ S_{4}(15) $ }
    }
  };
\end{tikzpicture}
}
% d006.tex
{\scalefont{\tikzpicturefont}
\begin{tikzpicture}[]
\node [circle,draw] (z){$ S_{3}(5) $}
  child 
  {
    node [circle,draw] (a) {$ S_{2}(2) $}
    child 
    { 
      [draw=white]
      node [draw=white] (b) {$  $}
    }
    child 
    { 
      node [circle,draw] (c) { $ S_{3}(4) $ }
    }
  }
  child 
  {
    node [circle,draw] (d) { $ S_{2}(7) $ } 
    child 
    {
      node [circle,draw] (e) { $ S_{3}(6) $ }
    }
    child 
    {
      node [circle,draw] (f) { $ S_{4}(15) $ }
    }
  };
\end{tikzpicture}
}
\caption*{Figure 2}
%\label{Subfloat label}
\end{figure}
\\
Continuing on, recall the algorithm had to construct the set $S_{3}(15)$ by executing
 six descending procedural calls to determine the base sets $S_{2}(7),\ldots ,S_{2}(2)$.
 As $S_{2}(2)$ is already present in the previous BST, this set is not re-constructed.
 Consequently in like manner to the above, a sequence of node insertions and re-balancing
 continues until all of the previous base sets have been inserted into the BST leading to,
 after another divisibility test, the insertion of the set $S_{3}(15)$. The resulting
 penultimate BST is shown on the left of Figure 3. Upon inserting the base set of $S_{2}(15)$
 to the left of $S_{3}(15)$ the final BST is pictured on the right of Figure 3.
 From this BST, the complete solution set $S(15)$ can now be by retrieved from
 those non-empty nodes contained within the list $\{S_{4}(15),S_{3}(15),S_{2}(15)\}$.\\
\vspace{2cm}

In implementing Algorithm 3.1, one could also have taken advantage of the linear
 complexity in searching and retrieval of such data structures as a Hash table,
 but this was not chosen due to the fact that Hash tables requires {\em a priori}
 estimates for the number of intermediary sets needed, which we do not have, and
 moreover BST are easier to implement in comparison. 
\\
\begin{figure}[htb]
\centering
% d011.tex
{\scalefont{\tikzpicturefont}
\begin{tikzpicture}[] 
\node [circle,draw] (z){$ S_{3}(5) $}
  child 
  {
    node [circle,draw] (a) {$ S_{3}(4) $}
    child 
    { 
      node [circle,draw] (b) {$ S_{2}(3)  $}
      child 
      { 
        node [circle,draw] (l) {$ S_{2}(2) $}
      }
      child 
      { 
        node [circle,draw] (k) {$ S_{2}(4)  $}
      }
    }
    child 
    { 
      node [circle,draw] (c) { $ S_{2}(5) $ }
    }
  }
  child 
  {
    node [circle,draw] (d) { $ S_{2}(7) $ } 
    child 
    {
      node [circle,draw] (e) { $ S_{3}(6) $ }
      child
      {
        node [circle,draw] (h) { $ S_{2}(6) $ }
      }
      child
      {
        [draw=white]
        node [draw=white] (i) { }
      }
    }
    child 
    {
      node [circle,draw] (f) { $ S_{4}(15) $ }
      child
      {
        node [circle,draw] (m) { $ S_{3}(15) $ }
      }
      child
      {
        [draw=white]
        node [draw=white] (n) { }
      }
    }
  };
\end{tikzpicture}
}
% d012.tex
{\scalefont{\tikzpicturefont}
\begin{tikzpicture}[] 
\node [circle,draw] (z){$ S_{3}(5) $}
  child 
  {
    node [circle,draw] (a) {$ S_{3}(4) $}
    child 
    { 
      node [circle,draw] (b) {$ S_{2}(3)  $}
      child 
      { 
        node [circle,draw] (l) {$ S_{2}(2) $}
      }
      child 
      { 
        node [circle,draw] (k) {$ S_{2}(4)  $}
      }
    }
    child 
    { 
      node [circle,draw] (c) { $ S_{2}(5) $ }
    }
  }
  child 
  {
    node [circle,draw] (d) { $ S_{2}(7) $ } 
    child 
    {
      node [circle,draw] (e) { $ S_{3}(6) $ }
      child
      {
        node [circle,draw] (h) { $ S_{2}(6) $ }
      }
      child
      {
        [draw=white]
        node [draw=white] (i) { }
      }
    }
    child 
    {
      node [circle,draw] (f) { $ S_{3}(15) $ }
      child
      {
        node [circle,draw] (m) { $ S_{2}(15) $ }
      }
      child
      {
        node [circle,draw] (n) { $ S_{4}(15) $ }
      }
    }
  };
\end{tikzpicture}
}
\\
\caption*{Figure 3}
%\label{Subfloat label}
\end{figure}
\\
To conclude we now address how the ``base" sets $S_{2}(\cdot)$
 were constructed within the C++ program.   
 From (\ref{eq:1ex}) it is clear that in order $S_{2}(n)$ can be
 determined for large $n>2$, one needs to perform a factorization of
 the form $n-1=d((n-1)/d)$, over all divisors $d$ of $n-1$, with
 $d\leq\sqrt{n-1}$. This was made practical via the use of the online
 number theory library LIDIA (see \cite{lidia}), which provided the
 necessary factorization algorithm, together with GMP (see \cite{gmp}), a
 free online library to perform the high precision integer arithmetic
 to compute and output the numerical values of the terms
 $\frac{n-1}{d}+1$ and $d+1$ necessary to construct each of elements of
 the set $S_{2}(n)$
\\
\section{Searching for Extra Exceptional Values}
As mentioned previously the Diophantine equation in (\ref{eq:2ex}) always has a solution for $n\geq 2$, as the ordered $n$-tuple $(2,n;n-2)$ satisfies the arithmetic identity
\[
\underbrace{1+1+\cdots +1}_{(n-2) 1's}+2+n=\underbrace{1\cdot 1\cdots 1}_{(n-2)1's}\cdot2\cdot n\mbox{ .}
\]
The ordered $n$-tuple $(2,n;n-2)$ was referred to in \cite{ecker} as the ``basic solution''. It has been noted in \cite{ecker},\cite{kurland} that for some values of $n\geq 2$, the basic solution is the only solution to the ESP-Problem in $n$ variables, that is where $S(n)=\{ (2,n;n-2)\}$. These special values of $n$ have been termed in \cite{ecker} as the Exceptional Values, of which $n=2$ is clearly one such value, via Lemma 2.1. From the structure of the solution set $S(n)$, one can see that for an integer $n>2$ to be an exceptional value, necessarily $S_{2}(n)=\{ (2,n;n-2)\}$. Now in view of (\ref{eq:1ex}), as the cardinality of $S_{2}(n)$ is equal to the number of ordered factorizations of $n-1=ab$, with $1\leq a\leq b$, we deduce that $S_{2}(n)=\{ (2,n;n-2)\}$ if and only if $n-1$ is prime. Thus when searching for exceptional values of $n>2$, it suffices to concentrate the search on those $n>2$ for which $n-1$ is prime. 
 Armed with this information the authors in \cite{ecker},\cite{kurland} via an exhaustive computer search of all primes
 less than $10^{10}$, uncovered the following list of exceptional values contained in the set $E=\{ 2,3,4,6,24,114,174,444\}$. As stated earlier, it is still an open conjecture as to whether the set $E$ is complete. In the following result, we provide a stricter necessary condition for an integer $n>2$ to be an element of $E$. Recall that a prime $p$ is a Sophie Germain prime if $2p+1$ is also a prime.
\begin{theo}
If an integer $n>2$ is an exceptional value of the ESP-Problem in $n$ variables, then $n-1$ must be a Sophie Germain prime number.
\end{theo}
\begin{proof}
If $n>2$ is an exceptional value, then necessarily $n-1$ is a prime number and the set $S_{3}(n)=\emptyset$. Recalling that the recursive generation of the set $S_{r}(n)$ is summarized in equation (\ref{eq:7}), observe after setting $r=3$ that
\begin{equation}
S_{3}(n)=\bigcup_{j=-1}^{\lfloor\frac{n-7}{2}\rfloor}\{ (x_{1},x_{2},w;n-3): (x_{1},x_{2};j+1)\in S_{2}(j+3), w=1+\frac{n-3-j}{x_{1}+x_{2}+j}\in\mathbb{N}\}\mbox{ .}\label{eq:10}
\end{equation}
Under assumption all the sets within the union must be empty, and so from (\ref{eq:10}) we deduce upon substituting the basic solution $(2,j+3;j+1)\in S_{2}(j+3)$ into the expression for $w$, that 
$w:=\frac{j+2+n}{2j+5}\not\in\mathbb{N}$, for $j=-1,\ldots ,\lfloor\frac{n-7}{2}\rfloor$. As $2j+5$ is an odd integer the previous conclusion further implies that $2w\not\in\mathbb{N}$, and so observe \begin{eqnarray}
2w:=\frac{2j+4+2n}{2j+5} & = & \frac{2j+5+(2n-1)}{2j+5} \nonumber\\
                             & = & 1 +\frac{2n-1}{2j+5}\not\in\mathbb{N}\mbox{ ,}\label{eq:11}
\end{eqnarray}
for $j=-1,\ldots ,\lfloor\frac{n-7}{2}\rfloor$. Recalling that $n-1$ is prime and so $n$ is an even integer, one finds that $\lfloor\frac{n-7}{2}\rfloor =\lfloor\frac{n-6}{2}-\frac{1}{2}\rfloor=\frac{n-6}{2}-1$, consequently $2j+5$ assumes the values of all odd integers $3,5,\ldots ,n-3\leq\sqrt{2n-1}<n$, for $j=-1,\ldots ,\lfloor\frac{n-7}{2}\rfloor$. As $n-1$ clearly does not divide $2n-1=2(n-1)+1$, we deduce from (\ref{eq:11}) that $2n-1=2(n-1)+1$ must be a prime number and so by definition $n-1$ is a Sophie Germain prime number.
\end{proof}  
 
 It is conjectured that the supply of Sophie Germain primes in infinite, and that the number of Sophie Germain primes less than or equal to $n$ is $O(\frac{n}{(\ln n)^{2}})$ (see \cite[p. 165]{riben}). Regardless of the validity of this conjecture, we can see that Theorem 1 together with the working program in \cite{program1}, provides us with a refined searching scheme for extra elements of $E$.
 Indeed all that is required 
 to take each known Sophie Germain prime number $p$ and input $n=p+1$ into \cite{}, if one non-basic solution is uncovered then terminate the algorithm, else we can conclude that $n=p+1$ is an extra exceptional value. To make this searching scheme practicable, one would need to have access to sufficiently large computing power, in view of the length of the largest known Sophie Germain prime number to date. As of April $2012$ the largest known Sophie Germain Prime number, discovered by P. Bliedung, using a distributed PrimeGrid search is, $18543637900515\times 2^{666667}-1$, and is comprised of $200701$ decimal digits (see \cite{blied}). Apart from this computational investigation, a final alternate perspective of Theorem 1 we may wish to consider is that, if it were possible to prove that the set $E$ was infinite, then the infinitude of Sophie Germain primes would
 follow at once. However the consensus of the authors is that this is very unlikely, as the Sophie Germain prime conjecture is a sufficiently deep problem of number theory, that such an elementary approach to settle this conjecture,
 would be at best wishful thinking. 

{\em RMIT University, GPO Box 2467V, Melbourne, Victoria 3001, Australia}\\
{\em michael.nyblom@rmit.edu.au}\\
{\em cheltonevans@gmail.com}
\end{document}